\documentclass[conference]{IEEEtran}
\usepackage{stmaryrd}
\usepackage{amsfonts}

\usepackage{times}
\usepackage{dsfont}
\usepackage{amsfonts,amsmath,amssymb}
\usepackage[usenames,dvipsnames]{pstricks}
\usepackage{epsfig}

\newtheorem{definition}{Definition}
\newtheorem{remark}[definition]{Remark}
\newtheorem{lemma}[definition]{Lemma}
\newtheorem{theorem}[definition]{Theorem}
\newtheorem{proposition}[definition]{Proposition}
\newtheorem{proof}[definition]{Proof}
\newtheorem{example}[definition]{Example}

\IEEEoverridecommandlockouts    

\begin{document}

\title{\ \\ \LARGE\bf Chaos of Protein Folding\thanks{Christophe Guyeux and Jacques M. Bahi
are with the Computer Science Laboratory LIFC, University of Franche-Comt\'e. Nathalie C\^ot\'e is a student in master of biology, same university.
16, route de Gray - 25030 Besan\c con, France (phone: +33 381
666948; email: \{christophe.guyeux, jacques.bahi\}@univ-fcomte.fr,nathalie.cote02@edu.univ-fcomte.fr).}}

\author{Jacques M. Bahi, Nathalie C\^ot\'e, and Christophe Guyeux*\thanks{* Authors in alphabetic order.}}


\maketitle

\begin{abstract}
As protein folding is a NP-complete problem, artificial intelligence tools like neural networks and genetic algorithms are used to attempt to predict the 3D shape of an amino acids sequence.
Underlying these attempts, it is supposed that this folding process is predictable.
However, to the best of our knowledge, this important assumption has been neither proven, nor studied. 
In this paper the topological dynamic of protein folding is evaluated.
It is mathematically established that protein folding in 2D hydrophobic-hydrophilic (HP) square lattice model is chaotic as defined by Devaney.
Consequences for both structure prediction and biology are then outlined.
\end{abstract}


\section{Introduction}

Proteins, formed by a string of amino acids folding into a specific tridimentional shape, carry out the majority of functionality within an organism.
However, simulating perfectly the folding processes or molecular dynamic occurring in biology nature is indeed infeasible, due to the following reasons.
Firstly, the forces involved in the stability of the protein conformation are currently not modeled with enough accuracy \cite{Hoque09}.
Indeed, we can even wonder if it is so realistic to hope finding one day an accurate model for this problem.
Secondly, due to an astronomically large number of possible 3D protein structures for a corresponding primary sequence of amino acids \cite{Hoque09}: the computation capability required even for handling a moderately-sized folding transition exceeds drastically the capacity of the most powerful computers around the world.

Consequently, proteins structures are not exactly computed, but they are predicted. 
As this Protein Structure Prediction (PSP) is a NP-complete problem \cite{Crescenzi98}, prediction for optimal protein structures is principally performed using computational intelligence approaches such as genetic algorithms, ant colonies \cite{Shmygelska2005Feb}, particle swarm, or neural networks.
Models of various resolutions are applied too, to tackle with the complexity of this problem.
In low resolution models, atoms into the same amino acid can for instance be considered as a same entity.
These low resolutions models are often used first to predict the backbone of the 3D conformation.
Then, high resolution models come next for further exploration. 
Such a prediction strategy is commonly used in PSP softwares like ROSETTA \cite{Bonneau01} or TASSER.

In this paper, we demonstrate that protein folding is indeed unpredictable, that is, it is chaotic according to Devaney.
This well-known topological notion for a chaotic behavior is one of the most established mathematical definition of unpredictability for dynamical systems.
This proof has been achieved in the framework of protein structure prediction for a 2D hydrophobic-hydrophilic (HP) lattice model \cite{Berger98}.
This popular lattice model with low resolution focuses only upon hydrophobicity by separating the amino acids into two sets: hydrophobic (H) and hydrophilic (or polar P) \cite{Dill1985}.
Numerous variations are proposed in the general HP model: 2D or 3D lattices, with square, cubic, triangular, or face-centered-cube shapes...
In our demonstration, we have chosen a 2D square lattice for easy understanding.
However the process still remains general and can be applied to high resolution models by using a more refined formulation. 

After having established the chaotic behavior of the folding dynamic by using two different proofs, we will outline the consequences of this fact.
More precisely, we will focus on the following questions.
Firstly, is it possible to predict 3D protein structures using artificial intelligence tools if the folding process is chaotic ? 
In other words, are genetic algorithms, neural networks, and so on, able to predict chaotic behaviors, at least as defined by Devaney (in this paper, we will only study neural networks) ?


The remainder of this paper is organized as follows. 
In Section \ref{Sec:basic recalls}, basic notations and terminologies concerning both HP-model and Devaney's topological chaos are recalled. 
Then in the next two sections the proofs of the chaotic behavior of protein folding dynamic are established.
The first proof is directly realized in the Devaney's context whereas the second one uses a previously proven result concerning chaotic iterations \cite{guyeux09}.
Consequences of this unpredictable behavior are outlined in Section \ref{Sec:Consequences}.
Among other things, it is regarded whether chaotic behaviors are harder to predict than ``normal'' behaviors or not.
Additionally, reasons explaining why a chaotic behavior unexpectedly leads to approximately one thousand categories of folds are proposed.
This paper ends by a conclusion section, in which our contribution is summarized and intended future work is presented.

\section{Basic Recalls}
\label{Sec:basic recalls}

In the sequel $S^{n}$ denotes the $n^{th}$ term of a sequence $S$ and $V_{i}$ the $i^{th}$ component of a vector $V$. 
$f^{k}=f\circ...\circ f$ is the $k^{th}$ composition of a function $f$.
Finally, the following notation is used: $\llbracket1;N\rrbracket =\{1,2,\hdots,N\}$.

\subsection{2D hydrophilic-hydrophobic (HP) model}

\subsubsection{The HP model}

In the HP model, hydrophobic interactions are supposed to dominate protein folding.
This model was formerly introduced by Dill, who consider in \cite{Dill1985} that the protein core freeing up energy is formed by hydrophobic amino acids, whereas hydrophilic amino acids tend to move in the outer surface due to their affinity with the solvent (see Figure \ref{fig:hpmodel}).
 
A protein conformation is then a self-avoiding walk (SAW) on a 2D or 3D lattice such that its energy $E$, depending on topological neighboring contacts between hydrophobic amino acids which are not contiguous in the primary structure, is minimal.
In other words, for an amino-acid sequence $P$ of length $\mathsf{N}$ and for the set $\mathcal{C}(P)$ of all SAW conformations of $P$, the chosen conformation will be $c^* = min \left\{E(c) \big/ c \in \mathcal{C}(P)\right\}$ \cite{Shmygelska05}.
In that context and for a conformation $c$, $E(c)=-q$ where $q$ is equal to the number of topological hydrophobic neighbors.
For example, $E(c)=-5$ in Figure \ref{fig:hpmodel}.

\begin{figure}[t]
\centering
\includegraphics[width=2.in]{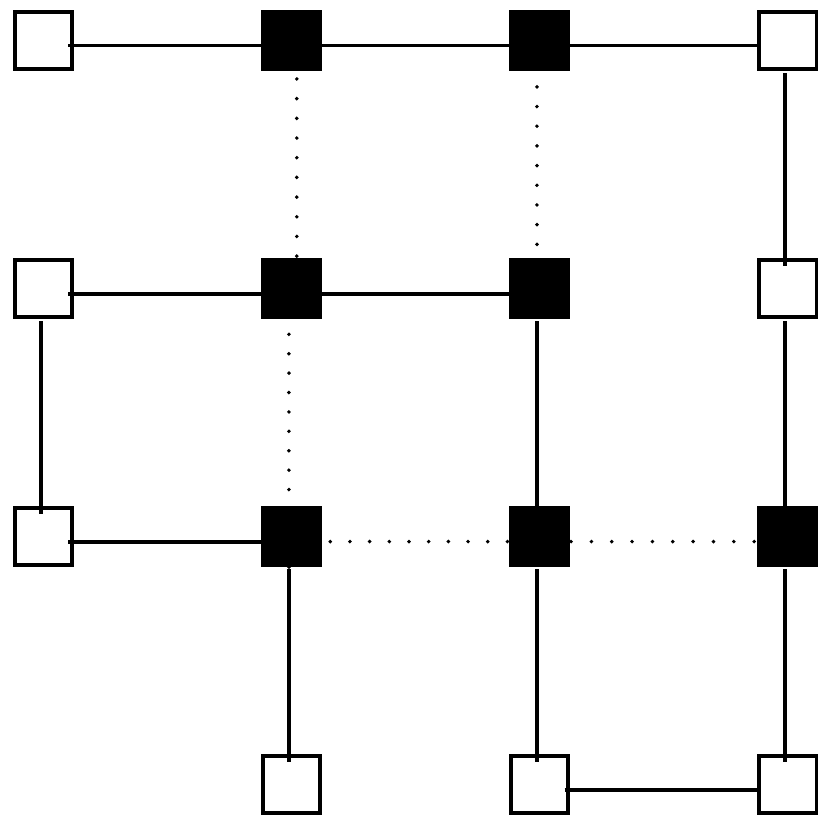}
\caption{Hydrophilic-hydrophobic (HP) model (black squares are hydrophobic residues).}
\label{fig:hpmodel}
\end{figure}

\subsubsection{Protein Encoding}

Additionally to the direct coordinate presentation, at least two other isomorphic encoding strategies for HP model are possible: relative encoding and absolute encoding.
In relative encoding \cite{Hoque09}, 
the move direction is defined relative to the direction of the previous move.
Alternatively, in absolute encoding \cite{Backofen99algorithmicapproach}, 
which is the encoding chosen in this paper, 
the direct coordinate presentation is replaced by letters or numbers representing directions with respect to the lattice structure. 

For a 2D absolute encoding, the permitted moves are: forward $\rightarrow$ (denoted by 1), backward $\leftarrow$ (2), up $\uparrow$ (3), and down $\downarrow$ (4).
A 2D conformation $c$ of $\mathsf{N}$ residues for a protein $P$ could then be $c \in \llbracket 1; 4 \rrbracket^{\mathsf{N}-2}$, as the initial move is always forward (1) \cite{Hoque09}.
For example, in Figure \ref{fig:hpmodel}, the 2D absolute encoding is (1)1144423322414 (starting from the upper left corner).
In that situation, $3^{\mathsf{N}-2}$ conformations are possible when considering $\mathsf{N}$ residues, even if some of them are invalid due to the SAW requirement.

\subsection{Devaney's Chaotic Dynamical Systems}
\label{subsection:Devaney}
\vspace{-4pt}

Let us now introduce the notion of chaos used in this document.
Consider a topological space $(\mathcal{X},\tau)$ and a continuous function $f$ on $\mathcal{X}$.

\begin{definition}
 $f$ is said to be \emph{topologically transitive} if, for any pair
 of open sets $U,V \subset \mathcal{X}$, there exists $k>0$ such that
 $f^k(U) \cap V \neq \emptyset$.
\end{definition}

\begin{definition}
 An element (a point) $x$ is a \emph{periodic element} (point) for
 $f$ of period $n\in \mathds{N}^*,$ if $f^{n}(x)=x$.
\end{definition}

\begin{definition}
 $f$ is said to be \emph{regular} on $(\mathcal{X}, \tau)$ if the set
 of periodic points for $f$ is dense in $\mathcal{X}$: for any point
 $x$ in $\mathcal{X}$, any neighborhood of $x$ contains at least one
 periodic point.
\end{definition}

\begin{definition}
 $f$ is said to be \emph{chaotic} on $(\mathcal{X},\tau)$ if $f$ is
 regular and topologically transitive.
\end{definition}

The chaos property is related to the notion of ``sensitivity'', defined on a metric space $(\mathcal{X},d)$ by:

\begin{definition}
 \label{sensitivity} $f$ has \emph{sensitive dependence on initial conditions}
 if there exists $\delta >0$ such that, for any $x\in \mathcal{X}$
 and any neighborhood $V$ of $x$, there exist $y\in V$ and
 $n \geq 0$ such that $d\left(f^{n}(x), f^{n}(y)\right) >\delta
 $.
\end{definition}

Indeed, Banks \emph{et al.} have proven in~\cite{Banks92} that when
$f$ is chaotic and $(\mathcal{X}, d)$ is a metric space, then $f$ has
the property of sensitive dependence on initial conditions (this
property was formerly an element of the definition of chaos). To sum
up, quoting Devaney in~\cite{Devaney}, a chaotic dynamical system ``is
unpredictable because of the sensitive dependence on initial
conditions. It cannot be broken down or simplified into two subsystems
which do not interact because of topological transitivity. And in the
midst of this random behavior, we nevertheless have an element of
regularity''. Fundamentally different behaviors are consequently
possible and occur in an unpredictable way.

\section{Protein Folding is Chaotic}
\label{sec:HP=Chaos}

We will now give a first proof of the chaotic behavior of the protein folding dynamic.

\subsection{Initial Premises}

Let us firstly introduce the preliminaries of our approach.

The primary structure of a given protein $p$ with $\mathsf{N}+1$ residues is coded by $1 1 \hdots 1$ ($\mathsf{N}$ times) in absolute encoding. 
Its final 2D conformation has an absolute encoding equal to $1 C_1^* \hdots C_{\mathsf{N}-1}^*$, where $\forall i, C_i^* \in \llbracket 1, 4 \rrbracket$, is such that $E(C^*) = min \left\{E(C) \big/ C \in \mathcal{C}(p)\right\}$.
This final conformation depends on the repartition of hydrophilic and hydrophobic amino acids in the initial sequence.

Moreover, we suppose that, if the residue number $n+1$ is forward the residue number $n$ in absolute encoding ($\rightarrow$) and if a fold occurs after $n$, then the forward move can only by changed into up ($\uparrow$) or down ($\downarrow$).
That means, in our simplistic model, only rotations of $+\frac{\pi}{2}$ or $-\frac{\pi}{2}$ are possible.

Obviously, for a given residue that is supposed to be updated, only one of the two possibilities below can appear for its absolute move during a fold:
\begin{itemize}
\item $1 \longmapsto 4, 4 \longmapsto 2, 2 \longmapsto 3,$ or $ 3 \longmapsto 1$ for a fold in the clockwise direction, or
\item $4 \longmapsto 1, 2 \longmapsto 4, 3 \longmapsto 2,$ or $1 \longmapsto 3$ for an anticlockwise. 
\end{itemize}

This fact leads to the following definition:

\begin{definition}
The \emph{clockwise fold function} is the function $f: \llbracket 1;4 \rrbracket \longrightarrow \llbracket 1;4 \rrbracket$ defined by:
$$
\begin{array}{crcl}
 & 1 & \longmapsto & 4 \\
 & 2 & \longmapsto & 3 \\
 & 3 & \longmapsto & 1 \\
 & 4 & \longmapsto & 2 \\
\end{array}
$$
\end{definition}

Obviously the anticlockwise fold function is $f^{-1}$. 

Thus at the $n^{th}$ folding time, a residue $k$ is chosen and its absolute move is changed by using either $f$ or $f^{-1}$. As a consequence, all of the absolute moves must be updated from the coordinate $k$ until the last one $\mathsf{N}$ by using the same folding function.

\begin{example}
\label{ex1}
If the current conformation is $C=111444$, i.e., $\rightarrow\rightarrow\rightarrow\downarrow\downarrow\downarrow$, and if the third residue is chosen to fold by a rotation of $-\frac{\pi}{2}$ (mapping $f$), thus the new conformation will be:
$$(C_1,C_2,f(C_3),f(C_4),f(C_5),f(C_6)) = (1,1,4,2,2,2).$$
\noindent That is, $\rightarrow\rightarrow\downarrow\leftarrow\leftarrow\leftarrow$.
\end{example}

These considerations lead to the formalization of the next section.

\subsection{Formalization and Notations}

Let $\mathsf{N}+1$ be a fixed number of amino acids, where $\in\mathds{N}^*$. 
We define $$\mathcal{X}=\llbracket 1; 4 \rrbracket^\mathsf{N}\times \llbracket -\mathsf{N};\mathsf{N} \rrbracket^\mathds{N}$$ as the phase space of the protein folding process.
An element $X=(C,F)$ of this dynamical folding space is constituted by:
\begin{itemize}
\item A conformation of the $\mathsf{N}+1$ residues in absolute encoding: $C=(C_1,\hdots, C_\mathsf{N}) \in \llbracket 1; 4 \rrbracket^\mathsf{N}$.
\item A sequence $F \in \llbracket -\mathsf{N} ; \mathsf{N} \rrbracket^\mathds{N}$ of future folds, depending on hydrophobicity, such that when $F_i \in \llbracket -\mathsf{N}; \mathsf{N} \rrbracket$ is $k$, it means that it occurs:
\begin{itemize}
\item a fold after the $k-$th residue by a rotation of $-\frac{\pi}{2}$ (mapping $f$) at the $i-$th step, if $k = F_i >0$,
\item no fold at time $i$ if $k=0$,
\item a fold after the $|k|-$th residue by a rotation of $\frac{\pi}{2}$ (\emph{i.e.}, $f^{-1}$) at the $i-$th time, if $k<0$.
\end{itemize}
\end{itemize}
On this phase space, the protein folding dynamic can be formalized as follows.

Denote by $i$ the map that transforms a folding sequence into its first term (the first folding operation): 
$$
\begin{array}{lccl}
i:& \llbracket -\mathsf{N};\mathsf{N} \rrbracket^\mathds{N} &\longrightarrow& \llbracket -\mathsf{N};\mathsf{N} \rrbracket,\\
& F &\longmapsto & F^0,
\end{array}$$
by $\sigma$ the shift function over $\llbracket -\mathsf{N};\mathsf{N} \rrbracket^\mathds{N}$, that is to say, 
$$
\begin{array}{lccl}
\sigma :& \llbracket -\mathsf{N};\mathsf{N} \rrbracket^\mathds{N} &\longrightarrow& \llbracket -\mathsf{N};\mathsf{N} \rrbracket^\mathds{N},\\
& \left(F^k\right)_{k \in \mathds{N}} &\longmapsto & \left(F^{k+1}\right)_{k \in \mathds{N}},
\end{array}$$
\noindent and by $sign$ the function:
$$
sign(x) = \left\{
\begin{array}{ll}
1 & \textrm{if } x>0,\\
0 & \textrm{if } x=0,\\
-1 & \textrm{else.}
\end{array}
\right.
$$

The shift function removes the first folding operation from the sequence after it has been achieved once.

Consider now the map $G:\mathcal{X} \to \mathcal{X}$ defined by:
$$G\left((C,F)\right) = \left( f_{i(F)}(C);\sigma(F)\right)$$
\noindent where $\forall k \in \llbracket -\mathsf{N};\mathsf{N} \rrbracket$, $f_k: \llbracket 1;4 \rrbracket^\mathsf{N} \to \llbracket 1;4 \rrbracket^\mathsf{N}$ is defined by: 
\begin{flushleft}
$f_k(C_1, ..., C_\mathsf{N}) =$
\end{flushleft}
\begin{flushright}
$ (C_1,... ,C_{|k|-1}, f^{sign(k)}(C_{|k|}),...,f^{sign(k)}(C_\mathsf{N})).$
\end{flushright}

Thus the folding process of a protein $P$ in the 2D HP square lattice model, with initial conformation equal to $(1,1, \hdots, 1)$ in absolute encoding, and a folding sequence equal to $(F^i)_{i \in \mathds{N}}$ provided by hydrophobic interactions, is defined by the following dynamical system over $\mathcal{X}$:
$$
\left\{
\begin{array}{l}
X^0 = ((1,1,\hdots,1);F)\\
X^{n+1} = G(X^n), \forall n \in \mathds{N}.
\end{array}
\right.
$$

In other words, at each step $n$, if $X^n=(C,F)$, we take the first folding operation to realize, that is $i(F) = F^0 \in \llbracket -\mathsf{N};\mathsf{N} \rrbracket$, we update the current conformation $C$ by rotating all of the residues coming after the $|i(F)|-$th one, which means that we replace the conformation $C$ with $f_{i(F)}(C)$.
Lastly, we remove this rotation (the first term $F^0$) from the folding sequence $F$: thus $F$ becomes $\sigma(F)$.

\begin{example}
Let us reconsider example \ref{ex1}. 
One iteration of its dynamical folding process can be described in this formalization by:
$\left((1,1,4,2,2,2),(F^1,F^2, \hdots)\right) =$ $G\left((1,1,1,4,4,4),(+3,F^1,F^2, \hdots)\right).$
\end{example}

\begin{remark}
A protein $P$ that has finished to fold, if such a protein exists, has the form $(C;(0,0,0,\hdots))$, where $C$ is the final 2D structure of $P$.
\end{remark}

\begin{remark}
Such a formalization allows the study of proteins that never stop to fold, for instance due to their never ended interactions with the environment.
\end{remark}

\subsection{A Metric for the Folding Process}

We define a metric $d$ over $\mathcal{X}=\llbracket 1;4 \rrbracket^\mathsf{N} \times \llbracket -\mathsf{N};\mathsf{N} \rrbracket^\mathds{N}$ by:
$$\displaystyle{d(X, \check{X}) = d_C(C, \check{C}) + d_F (F, \check{F}).}$$
\noindent where $\delta(a,b) =0$ if $a=b$, else $\delta(a,b)=1$, and
$$
\left\{
\begin{array}{ll}
d_C(C, \check{C}) = & \displaystyle{\sum_{k=1}^\mathsf{N} \delta(C_k,\check{C}_k) 2^{\mathsf{N}-k}} \\
d_F (F, \check{F}) = & \displaystyle{\dfrac{9}{2 \mathsf{N}} \sum_{k=1}^\infty \dfrac{|F^k-\check{F}^k|}{10^k}}
\end{array}
\right.$$

This new distance for the dynamical description of the protein folding process in 2D HP square lattice model can be justified as follows.

The integral part of the distance between two points $X=(C,F)$ and $\check{X}=(\check{C},\check{F})$ of $\mathcal{X}$ measures the differences between the current 2D structures of $X$ and $\check{X}$. 
More precisely, if $d_C(C,\check{C})$ is in $\llbracket 2^k,2^{k+1} \rrbracket$, then the first $k$ terms in the conformations $C$ and $C'$ (absolute encoding) are equal, whereas the $k+1^{th}$ terms differ.
The decimal part of $d(X,\check{X})$ will decrease when the duration the folding process will be similar increase.
More precisely, $F^k = \check{F}^k$ if and only if the $k+1^{th}$ digit of this decimal part is 0. 
Lastly, $\frac{9}{\mathsf{N}}$ is just a normalization factor.

For instance, if we know where are the $\mathsf{N}+1$ residues of our protein $P$ in the lattice, and if we know what will be its $k$ next folding, then we are into the ball $\mathcal{B}(C,10^{-k})$, that is, very close to the point $(C,F)$ if $k$ is large.

\begin{remark}
In $X^0 = ((1,1, \hdots, 1), F)$, the folding sequence $F^0$ results, among other things, on the hydrophobic interactions between amino acids. 
Indeed, it is this $F^0$ that is searched, when trying to predict the folding process of a given protein $P$.
That is to say, the error on $X^0$ measured by $d$ corresponds to our incapacity to determine \emph{exactly} the whole future folding process $F$ of $P$.
Improving this prediction with better computational intelligence tools leads to the reduction of the distance between $X^0$ (what is looked for) and ${X'}^0$ (our approximation).

The question raised by this study is: even if we cannot have access with an infinite precision to all of the forces that participate to the folding process, \emph{i.e.}, even if we only know an approximation ${X'}^0$ of $X^0$, can we claim that the predicted conformation ${X'}^{n_1}$ still remains close to the true conformation ${X}^{n_2}$ ?
Or, on the contrary, do we have a chaotic behavior, a kind of butterfly effect that magnifies any error on the evaluation of the forces in presence ?
\end{remark}

Raising such a question leads to the study of the dynamical behavior of the folding process. 
To do so, we must firstly establish that $G$ is a continuous map on $(\mathcal{X},d)$.

\subsection{Continuity of the Folding Operation}

\begin{theorem}
$G:\mathcal{X} \to \mathcal{X}$ is a continuous map.
\end{theorem}

\begin{proof}
We will use the sequential characterization of the continuity.
Let $(X^n)_{n \in \mathds{N}} = \left((C^n,F^n)\right)_{n \in \mathds{N}} \in \mathcal{X}^\mathds{N},$ such that $X^n \rightarrow X = (\check{C},\check{F})$. 
We will show that $G\left(X^n\right) \rightarrow G(X)$.
Let us firstly remark that $\forall n \in \mathds{N}, F^n$ is a sequence: $F$ is thus a sequence of sequences.

On the one hand, as $X^n=(C^n,F^n) \rightarrow (\check{C},\check{F})$, we have $d_C\left(C^n,\check{C}\right) \rightarrow 0$, thus $\exists n_0 \in \mathds{N},$ $n \geqslant n_0$ $\Rightarrow d_C(C^n,\check{C})=0$. That is, $\forall n \geqslant n_0,$ $\forall k \in \llbracket 1;\mathsf{N} \rrbracket$, $\delta(C_k^n,C_k) = 0$, and so $C^n = \check{C}, \forall n \geqslant n_0.$ 
Additionally, $d_F(F^n,\check{F}) \rightarrow 0$, then $\exists n_1 \in \mathds{N},$ $d_F(F^n, \check{F}) \leqslant \frac{1}{10}$. As a consequence, $\exists n_1 \in \mathds{N},$ $\forall n \geqslant n_1$, the first term of the sequence $F^n$ is $\check{F}^0$: $i(F^n) = i(\check{F})$.
So, $\forall n \geqslant max(n_0,n_1),$ $f_{i(F^n)}\left(C^n\right) = f_{i\left(\check{F}\right)}\left(\check{C}\right)$, and then $f_{i(F^n)}\left(C^n\right)$ $\rightarrow$ $f_{i\left(\check{F}\right)}\left(\check{C}\right)$.

On the other hand, $\sigma(F^n) \rightarrow \sigma(F)$.
Indeed, $F^n \rightarrow F$ implies  $\sum_{k=1}^{\infty} \frac{| \left(F^n\right)^k-\check{F}^k |}{10^k} \rightarrow 0$.
Thus, $\frac{1}{10} \sum_{k=1}^{\infty} \frac{| \left(F^n\right)^{k+1}-\check{F}^{k+1} |}{10^k} \rightarrow 0$, so $\sum_{k=1}^{\infty} \frac{| \sigma(F^n)^k-\sigma(\check{F})^k |}{10^k} \rightarrow 0$.
Finally, $\sigma(F^n) \rightarrow \sigma(\check{F}).$

To conclude, as $f_{i(F^n)}\left(C^n\right)$ $\rightarrow$ $f_{i\left(\check{F}\right)}\left(\check{C}\right)$ and $\sigma(F^n) \rightarrow \sigma(\check{F})$, we have $G\left(X^n\right) \rightarrow G(X)$.
\end{proof}

It is now possible to study the chaotic behavior of the folding process.

\subsection{Regularity of the Folding Operation}

Let us firstly introduce the following notations: for $X = (C,F) \in \llbracket 1;4 \rrbracket^\mathsf{N}\times \llbracket 1;\mathsf{N} \rrbracket^\mathds{N}$, $\mathcal{C}(X) = C$ and $\mathcal{F}(X)=F$. We will now prove that,

\begin{lemma}
\label{lem}
For all $C,C'$ in $\llbracket 1;4 \rrbracket^\mathsf{N},$ there exist $ k_1, \hdots, k_\mathsf{N}$ in $\llbracket -\mathsf{N}; \mathsf{N} \rrbracket$ s.t. $G^\mathsf{N}\left(C,(k_1, \hdots, k_\mathsf{N},0,\hdots)\right) = \left(C',(0, 0, \hdots ) \right).$
\end{lemma}

\begin{proof}
We will prove this lemma by a mathematical induction on $\mathsf{N} \in \mathds{N}^*$.

For the base case $\mathsf{N}=1$, if $C_1 = C_1'$, then the result is satisfied with $k_1=0$. Else, either $C_1' = f_1(C_1)$ then $k_1=1$ holds, or $C_1' = f_1^{-1}(C_1)$ then $k_1=-1$.

Let us now suppose that the statement holds for some $\mathsf{N} \in \mathds{N}^*$. Let $C,C' \in \llbracket 1;4 \rrbracket^{\mathsf{N}+1}$. According to the inductive hypothesis, $\exists k_1, \hdots, k_\mathsf{N} \in \llbracket -\mathsf{N}, \mathsf{N} \rrbracket$ such that $\tilde{G}^{\mathsf{N}}\left( (C_1, \hdots, C_\mathsf{N}), (k_1, \hdots, k_\mathsf{N},0, \hdots) \right) = \left( (C_1', \hdots, C_\mathsf{N}'), (0, 0,\hdots) \right)$, where $\tilde{G}$ is the restriction of $G$ on its $\mathsf{N}-$th firsts variables. Let $x = \mathcal{C}\left(G^\mathsf{N}\left( (C_1, \hdots, C_\mathsf{N+1}), (k_1, \hdots, k_\mathsf{N+1},0, \hdots) \right)\right)_{\mathsf{N+1}}$. If $x=C_{\mathsf{N}+1}'$, then $k_{\mathsf{N}+1}=0$ holds. Else, either $f(x)=C_{\mathsf{N}+1}'$, and $k_{\mathsf{N}+1}=\mathsf{N}+1$ holds,
or $f^{-1}(x)=C_{\mathsf{N}+1}'$, and then $k_{\mathsf{N}+1}=-(\mathsf{N}+1$).
\end{proof}

We can now prove that,

\begin{proposition}
Protein folding is regular.
\end{proposition}

\begin{proof}
Let $X=(C,F) \in \mathcal{X}$ and $\varepsilon > 0$. 
Define $k_0=-\lfloor log_{10} (\varepsilon) \rfloor$
and $\tilde{X}$ such that:
(1) $\mathcal{C}(\tilde{X}) = C$,
(2) $\forall k \leqslant k_0, \mathcal{F}\left(G^k(\tilde{X})\right) = \mathcal{F}\left(G^k(X)\right)$,
(3) $\forall i \in \llbracket 1; \mathsf{N} \rrbracket, \mathcal{F}\left(G^{k_0+i}(\tilde{X})\right) = k_i,$ and
(4) $\forall i \in \mathds{N}, \mathcal{F} \left(G^{k_0+\mathsf{N}+i+1}(\tilde{X})\right) = \mathcal{F}\left(G^i(\tilde{X})\right)$,
 where $k_1, \hdots, k_n$ are integers given by lemma \ref{lem} with $C=\mathcal{C}\left( G^{k_0}(X) \right), C'=\mathcal{C}(X)$. Such a $\tilde{X}$ is a periodic point for $G$ into the ball $\mathcal{B}(X,\varepsilon)$. 
(1) and (2) are to make $\tilde{X}$ $\varepsilon-$close to $X$, (3) is for mapping $\mathcal{C}\left(G^{k_0}(\tilde{X})\right)$ into $C$ in at most $\mathsf{N}$ folding process. 
Lastly, (4) is for the periodicity of the folding process.
\end{proof}

\subsection{Transitivity of the Folding Operation}

Instead of proving the transitivity of $G$, which is required in the definition of chaos, we will establish its strong transitivity:

\begin{definition}
A dynamical system $\left( \mathcal{X}, f\right)$ is strongly transitive if $\forall x,y \in \mathcal{X},$ $\forall r > 0,$ $\exists z \in \mathcal{X},$ $d(z,x) \leqslant r \Rightarrow$ $\exists n \in \mathds{N}^*,$ $f^n(z)=y$.
\end{definition}

Obviously, strong transitivity implies transitivity. Let us now prove that,

\begin{proposition}
Protein folding is strongly transitive.
\end{proposition}

\begin{proof}
Let $X_A=(C_A,F_A)$, $X_B=(C_B, F_B)$ and $\varepsilon > 0$.
We will show that $X \in \mathcal{B}\left(X_A, \varepsilon\right)$ and $n \in \mathds{N}$ can be found such that $G^n(X)=X_B$.
Let $k_0 = - \lfloor log_{10} (\varepsilon ) \rfloor$ and $\check{X}=G^{k_0}(C_A,F_A)$, denoted by $\check{X}=(\check{C},\check{F})$.
According to lemma \ref{lem} applied to $\check{C}$ and $C_B$, $\exists k_1, \hdots, k_\mathsf{N} \in \llbracket -\mathsf{N}, \mathsf{N} \rrbracket$ such that $$G^\mathsf{N}\left( \check{C}, (k_1, \hdots, k_\mathsf{N},0,\hdots)\right) = \left(C_B, (0, \hdots )\right).$$
Let us define $X=(C,F)$ in the following way:
(1) $C=C_A$,
(2) $\forall k \leqslant k_0, F^k=F_A^k$,
(3) $\forall i \in \llbracket 1; \mathsf{N} \rrbracket, F^{k_0+i} = k_i$, and
(4) $\forall i \in \mathds{N}, F^{k_0+\mathsf{N}+i+1}=F_B^i$.
This point $X$ is thus an element of $\mathcal{B}(X_A,\varepsilon)$ (due to $1,2$), which is such that $G^{k_0+\mathsf{N}+1}(X) = X_B$ (by using $3,4$). As a consequence, $G$ is strongly transitive.
\end{proof}

This property is very important. It shows among other things that being as close as possible of the true folding process, for instance by using a very large basis of knowledge and numerous levels of resolution, is not a guarantee of success. Indeed, \emph{for any possible conformation} $c$, there is a prediction as good as possible of our considered protein, which leads to $c$.

\subsection{Chaotic behavior of the folding process}

As $G$ is regular and (strongly) transitive, we have:

\begin{theorem}
The protein folding process $G$ is chaotic according to Devaney.
\end{theorem}

Consequently this process is highly sensitive to its initial condition. 
In particular, even a minute difference on an intermediate conformation of the protein, in forces that act in the folding process, or in the position of an atom, can lead to enormous differences in its final conformation, even over fairly small timescales.
This is the so-called butterfly effect.
In particular, it seems very difficult to predict the 2D structure of a given protein by using the knowledge of the structure of similar proteins.
Let us finally remark that the whole 3D folding process with real torsion angles is obviously more complex than this 2D HP model.
Indeed, if the complete 3D folding process were predictable, then this simplistic version would be predictable too, as one of its particular cases.

As a conclusion, theoretically speaking, the folding process in unpredictable. Before studying some practical aspects of this unpredictability in Section \ref{Sec:Consequences}, we will initiate a second proof of the chaotic behavior of this process.

\section{Outlines of a second proof}

\subsection{Motivations}

In this section a second proof of the chaotic behavior of the protein folding process is given.
It is proven that the folding dynamic can be modeled as chaotic iterations (CIs). CIs are a tool used in distributed computing and in the computer science security field \cite{guyeuxTaiwan10} that has been established to be chaotic according to Devaney \cite{guyeux10}. 

This second proof is the occasion to introduce these CIs, which will be used at the end of this paper to study whether a chaotic behavior is really more difficult to learn with a neural network than a ``normal'' behavior.

\subsection{Chaotic Iterations: Basic Recalls}

Let us consider a \emph{system} with a finite number $\mathsf{N} \in
\mathds{N}^*$ of elements (or \emph{cells}), so that each cell has a
boolean \emph{state}. 
A sequence of length $\mathsf{N}$ of boolean
states of the cells corresponds to a particular \emph{state of the
system}. 
A sequence, which elements are subsets of $\llbracket 1;\mathsf{N}
\rrbracket $, is called a \emph{strategy}. The set of all strategies is denoted by $\mathbb{S}$ and the set $\mathds{B}$ is for $\{0,1\}$.

\begin{definition}
\label{Def:chaotic iterations}
Let
$f:\mathds{B}^{\mathsf{N}}\longrightarrow \mathds{B}^{\mathsf{N}}$ be
a function and $S\in \mathbb{S}$ be a strategy. The so-called
\emph{chaotic iterations} (CIs) are defined by $x^0\in
\mathds{B}^{\mathsf{N}}$ and $\forall n\in \mathds{N}^{\ast }, \forall i\in
\llbracket1;\mathsf{N}\rrbracket ,$
$$ 
x_i^n=\left\{
\begin{array}{ll}
 x_i^{n-1} & \text{ if } i \notin S^n \\ 
 \left(f(x^{n-1})\right)_{S^n} & \text{ if } i \in S^n.
\end{array}\right.
$$
\end{definition}

In other words, at the $n^{th}$ iteration, only the $S^{n}-$th cells are
\textquotedblleft iterated\textquotedblright .
Let us remark that the term ``chaotic'', in the name of these iterations, has \emph{a
priori} no link with the mathematical theory of chaos recalled previously.

\medskip

We will now recall that CIs can be written as a dynamical system, and characterize functions $f$ such that their CIs are chaotic according to Devaney \cite{guyeux09}.

\subsection{CIs and Devaney's chaos}

Given a function 
$f: \mathds{B}^{\mathsf{N}}\longrightarrow \mathds{B}^{\mathsf{N}}$, 
define the function $F_{f}:$ 
$\llbracket1;\mathsf{N}\rrbracket\times \mathds{B}^{\mathsf{N}}\longrightarrow 
\mathds{B}^{\mathsf{N}}$ 
such that
\begin{equation*}
F_{f}(k,E)=\left( E_{j}.\delta (k,j)+f(E)_{k}.\overline{\delta (k,j)}\right)_{j\in \llbracket1;\mathsf{N}\rrbracket},
\end{equation*}

\noindent where + and . are the boolean addition and product
operations, $\overline{x}$ is for the negation of $x$.

We have proven in \cite{guyeux09} that chaotic iterations can be described by the following
dynamical system:
\begin{equation*}
\left\{
\begin{array}{l}
X^{0}\in \tilde{\mathcal{X}} \\
X^{k+1}=\tilde{G}_{f}(X^{k}).
\end{array}
\right.
\end{equation*}

\noindent where $\tilde{G}_{f}\left( S,E\right) =\left( \sigma (S),F_{f}(i(S),E)\right)$, and $\tilde{\mathcal{X}}$ is a metric space for an ad hoc distance such that $\tilde{G}$ is continuous on $\mathcal{X}$ \cite{guyeux09}.

\medskip

Let now be given a configuration $x$. In what follows the configuration
$N(i,x) = (x_1,\ldots,\overline{x_i},\ldots,x_n)$ is obtained by
switching the $i-$th component of $x$. Intuitively, $x$ and $N(i,x)$
are neighbors. The chaotic iterations of the function $f$ can be
represented by the graph $\Gamma(f)$ defined below.

\begin{definition}
In the oriented \emph{graph of iterations} $\Gamma(f)$, vertices are
configurations of $\mathds{B}^\mathsf{N}$, and there is an arc labeled
$i$ from $x$ to $N(i,x)$ iff $F_f(i,x)$ is $N(i,x)$.
\end{definition}

We have proven in \cite{GuyeuxThese10} that:

\begin{theorem}
\label{Th:Caracterisation des IC chaotiques} Functions $f :
\mathds{B}^{n} \to \mathds{B}^{n}$ such that $\tilde{G}_f$ is chaotic
according to Devaney, are functions such that the graph $\Gamma(f)$ is
strongly connected.
\end{theorem}


We will now show that the protein folding process can be modeled as chaotic iterations, and conclude the proof by using the theorem recalled above.

\subsection{Protein Folding as Chaotic Iterations}

The attempt to use chaotic iterations with a view to model protein folding can be justified as follows.
At each iteration, the same process is applied to the system (\emph{i.e.}, to the conformation), that is the folding operation.
Additionally, it is not a necessity that all of the residues fold at each iteration:
indeed it is possible that, at a given iteration, only some of these residues folds.
Such iterations, where not all the cells of the considered system are to be updated, are exactly the iterations modeled by CIs.

Indeed, the protein folding process with folding sequence $(F^n)_{n \in \mathds{N}}$ consists in the following chaotic iterations: $C^0 = (1,1, \hdots, 1)$ and,
$$
C_{|i|}^{n+1} = \left\{
\begin{array}{ll}
C_{|i|}^n & \textrm{if } i \notin S^n,\\
f^{sign(i)}(C^n)_i & \textrm{else},
\end{array}
\right.
$$
\noindent where the chaotic strategy is defined by $\forall n \in \mathds{N}, S^n = \llbracket -\mathsf{N}; \mathsf{N} \rrbracket \setminus \llbracket -F^n; F^n \rrbracket$.

Thus, to prove that the protein folding process is chaotic as defined by Devaney, is equivalent to prove that the graph of iterations of the CIs defined above is strongly connected. This last fact is obvious, as it is always possible to find a folding process that map any conformation $(c_1, \hdots, c_\mathsf{N}) \in \llbracket 1 ;4 \rrbracket^\mathsf{N}$ to any other $(c_1', \hdots, c_\mathsf{N}') \in \llbracket 1 ;4 \rrbracket^\mathsf{N}$ (this is lemma \ref{lem}).

Let us finally remark that it is easy to study processes s.t. more than one fold occur per time unit, by using CIs.
This point will be deepened in a future work.
We will now investigate some consequences of the chaotic behavior of the folding process.

\section{Consequences}
\label{Sec:Consequences}

\subsection{Is a chaotic behavior incompatible with approximately one thousand folds ?}

Claiming that the protein folding process is chaotic seems to be contradictory with the fact that only approximately one thousand folds have been discovered this last decade. 
The number of proteins that have an understood 3D structure increase largely year after year. 
However the number of new categories of folds seems to be limited by a fixed value approximately equal to one thousand.
Indeed, there is no contradiction as a chaotic behavior does not forbid a certain form of order.
For example, seasons are not forbidden even if weather forecast has a non-intense chaotic behavior.
A same regularity appears in brains: even if hazard and chaos play an important rule in a microscopic scale, a statistical order appears in the complete neural network.
 
That is, a certain order can emerge from a chaotic behavior, even if it is not a rule of thumb.
More precisely, in our opinion these thousand folds can be related to basins of attractions or strange attractors of the dynamical system, objects that are well described by the mathematical theory of chaos.
Thus, it should be possible to determine all of the folds that can occur, by refining our model and looking for its basins of attractions with topological tools.
However, this assumption still remains to be more largely investigated.

\subsection{Is Artificial Intelligence able to Predict Chaotic dynamic ?}

\subsubsection{Experimental Protocol}

We will now wonder whether a chaotic behavior can be learned by a neural network or not. 
These considerations have been formerly proposed in \cite{arxivRNNchaos}.

We consider $f:\mathds{B}^\mathsf{N} \longrightarrow \mathds{N}^\mathsf{N}$, strategies of singletons ($\forall n \in \mathds{N}, S^n \in \llbracket 1; \mathsf{N} \rrbracket$), and a MLP which recognize $F_{f}$. 
That means, for all $(k,x) \in \llbracket 1 ; \mathsf{N} \rrbracket \times
\mathds{B}^\mathsf{N}$, the response of the output layer to the input
$(k,x)$ is $F_{f}(k,x)$.
We thus connect the output layer to the input one as it is depicted in
Figure~\ref{perceptron}, leading to a global recurrent artificial neural network (ANN) working as follows \cite{arxivRNNchaos}.

\begin{figure}[!t]
\centering
\includegraphics[width=2.75in]{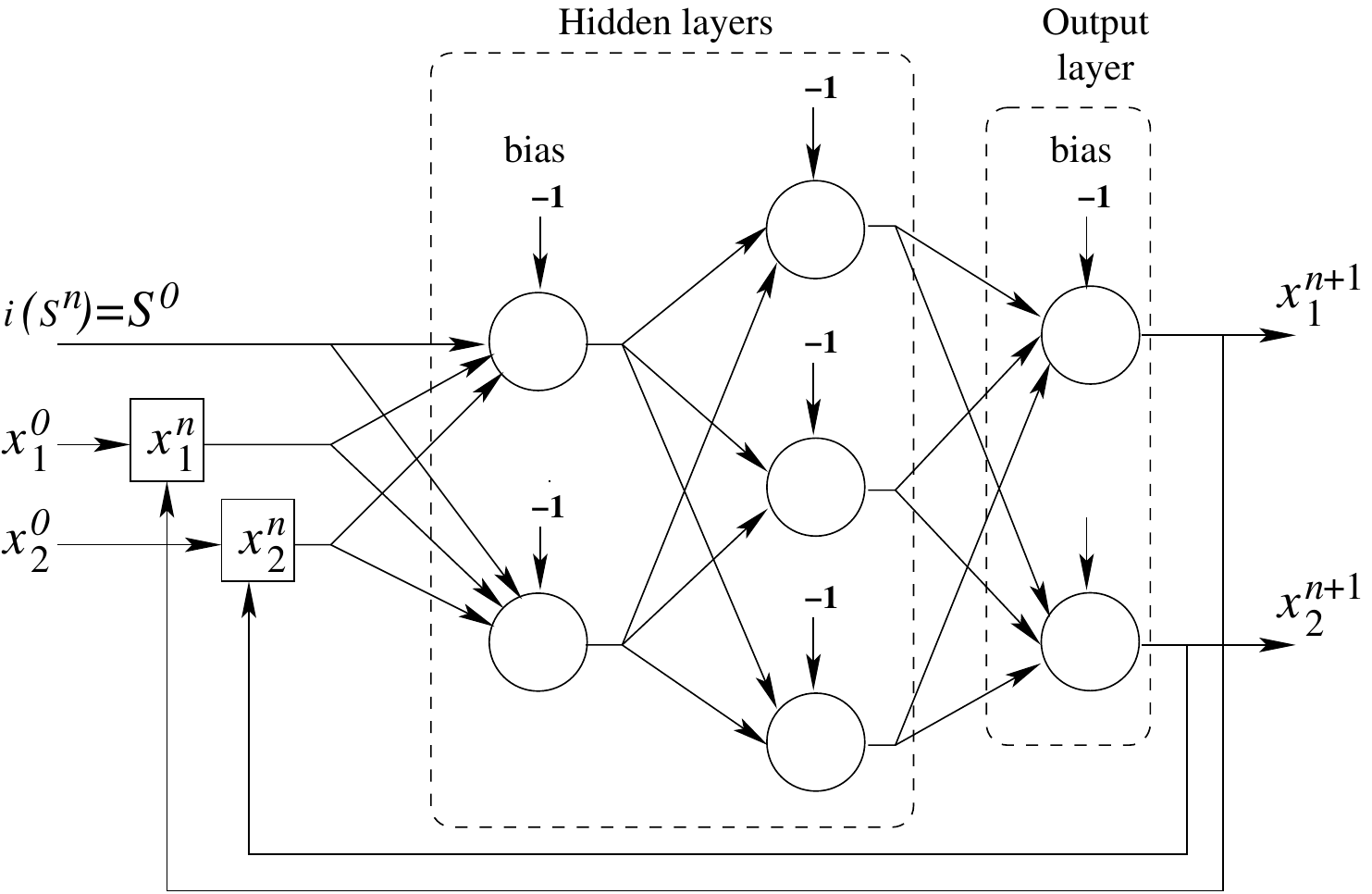}
\caption{RNN modeling $F_{f}$}
\label{perceptron}
\hfil
\end{figure}

At the initialization stage, the ANN receives a boolean vector
$x^0\in\mathds{B}^\mathsf{N}$ as input state, and $S^0 \in \llbracket
1;\mathsf{N}\rrbracket$ in its input integer channel $i()$. Thus,
$x^1 = F_{f}(S^0, x^0)\in\mathds{B}^\mathsf{N}$ is computed by the
neural network. This state $x^1$ is published as an output. Additionally, $x^1$
is sent back to the input layer, to act as boolean state in the next
iteration.
Finally, at iteration number $n$, the recurrent neural network receives
the state $x^n\in\mathds{B}^\mathsf{N}$ from its output layer and
$i\left(S^n\right) \in \llbracket 1;\mathsf{N}\rrbracket$ from its
input integer channel $i()$. It can thus calculate $x^{n+1} =
F_{f}(i\left(S^n\right), x^n)\in\mathds{B}^\mathsf{N}$, which will
be the new output of the network.
Obviously, this particular MLP produces exactly the same values than CIs with update function $f$. 
That is, such MLPs are equivalent, when working with $i(s)$, to CIs with $f$ as update function \cite{arxivRNNchaos} and strategy $S$.

Let us now introduce the two following functions:
$f_1(x_{1},x_2,x_3)=(\overline{x_{1}},\overline{x_{2}},\overline{x_{3}})$ and
 $f_2(x_{1},x_2,x_3)=(\overline{x_{1}},x_{1},x_{2})$.
It can easily be checked that these functions satisfy the hypothesis of Theorem \ref{Th:Caracterisation des IC chaotiques}, thus their CIs are chaotic according to Devaney. 
Then when the MLP defined above learn to recognize $F_{f_1}$ or $F_{f_2}$, indeed it tries to learn these CIs, that is, a chaotic behavior as defined by Devaney \cite{arxivRNNchaos}.
On the contrary, the function $g(x_{1},x_2,x_3)=(\overline{x_{1}},x_{2},x_{3})$
is such that $\Gamma(g_1)$ is not strongly connected. In this case, due to Theorem \ref{Th:Caracterisation des IC chaotiques}, the MLP does not learn a chaotic process.
We will now study the training process of functions $F_{f_1}$, $F_{f_2}$, and $F_{g}$ \cite{arxivRNNchaos}, that is to say, the ability to learn one iteration of CIs.

\subsubsection{Experimental results}

For each neural network we have considered MLP architectures with one and two hidden layers, with in the first case different numbers of hidden neurons (sigmoidal activation). 
Thus we will have different versions of a neural network modeling the same iteration function \cite{arxivRNNchaos}.
Only the size and number of hidden layers may change, since the numbers of inputs and output neurons (linear activation) are fully specified by the function. 
The neural networks are trained using the quasi-Newton L-BFGS (Limited-memory Broyden-Fletcher-Goldfarb-Shanno) algorithm in combination with the Wolfe linear search \cite{arxivRNNchaos}. 
The training is performed until the learning error (MSE) is lower than a chosen threshold value ($10^{-2}$).

\begin{table}[!t]
\renewcommand{\arraystretch}{1.3}
\caption{Results of some iteration functions learning, using different recurrent MLP architectures}
\label{results}
\centering
\begin{scriptsize}
\begin{tabular}{|c||c|c||c|c|}
\hline 
 & \multicolumn{4}{c|}{One hidden layer} \\
\cline{2-5}
 & \multicolumn{2}{c||}{8 neurons} & \multicolumn{2}{|c|}{10 neurons} \\
\hline
Function & Mean & Success & Mean & Success \\
 & epoch & rate & epoch & rate \\
\hline
$f_1$ & 82.21 & 100\% & 73.44 & 100\% \\
$f_2$ & 76.88 & 100\% & 59.84 & 100\% \\
$g_1$ & 36.24 & 100\% & 37.04 & 100\% \\
\hline
\hline
 & \multicolumn{4}{c|}{Two hidden layers: 8 and 4 neurons} \\
\cline{2-5}
 & \multicolumn{2}{c|}{Mean epoch number} & \multicolumn{2}{|c|}{Success rate} \\
\hline
$f_1$ & \multicolumn{2}{c|}{203.68} & \multicolumn{2}{c|}{76\%} \\
$f_2$ & \multicolumn{2}{c|}{135.54} & \multicolumn{2}{c|}{96\%} \\
$g_1$ & \multicolumn{2}{c|}{72.56} & \multicolumn{2}{c|}{100\%} \\
\hline
\end{tabular}
\end{scriptsize}
\end{table}

Table~\ref{results} gives for each considered neural network the mean number of epochs needed to learn one iteration in their ICs, and a success rate which reflects a successful training in less than 1000 epochs. 
Both values are computed considering 25 trainings with random weights and biases initialization. 
These results highlight several points \cite{arxivRNNchaos}. 
Firstly, the two hidden layer structure seems to be quite inadequate to learn chaotic behaviors.
Secondly, training networks so that they behave chaotically seems to be difficult for these simplistic functions only iterated one time, since they need in average more epochs to be correctly trained. 
However, the correctness of this point needs to be further investigated. 

At this point we can only claim that it is not completely evident that computational intelligence as neural networks are able to predict, with a good accuracy, protein folding. 
To reinforce this belief, tools optimized to chaotic behaviors must be found -- if such tools exist.
Similarly, there should be a link between the training difficulty and the ``quality'' of the disorder induced by a chaotic iteration function (their constants of sensitivity, expansivity, etc.), and this link must be found.

\section{Conclusion}
\label{Conclusion}\vspace{-4pt}

In this paper the topological dynamic of protein folding is evaluated.
More precisely, it is regarded whether this folding process is predictable or not.
It is achieved to determine if it is reasonable to think that computational intelligence as neural networks are able to predict the 3D shape of an amino acids sequence.
It is mathematically proven, by using two different ways, that protein folding in 2D hydrophobic-hydrophilic (HP) square lattice model is chaotic according to Devaney.

Consequences both for structure prediction and biology are then outlined.
In particular, the first comparison of the learning by neural networks of a chaotic behavior on the one hand, and of a more natural dynamic on the other hand, are outlined.
Obtained results tend to show that such chaotic behaviors are more difficult to learn than non-chaotic ones.
It is not our pretension to claim that it is impossible to predict chaotic behaviors as protein folding with computational intelligence. 
Our opinion is just that this important point must now be regarded with attention.

In future work the dynamical behavior of the protein folding process will be more deeply studied, by using topological tools as expansivity, topological mixing, Knudsen and Li-Yorke notions of chaos, topological entropy, etc.
The quality and intensity of this chaotic behavior will then be evaluated.
Consequences both on folding prediction and on biology will then be regarded in detail.
Other molecular or genetic dynamics will be investigate by using mathematical topology, and other chaotic behaviors will be looked for (as neurons in the brain).
More specifically, various tools taken from the field of computational intelligence will be studied to determine if some of these tools are capable to predict behaviors that are chaotic.
It is highly possible that prediction depends both on the tool and on the chaos quality.
Moreover, the study presented in this paper will be extended to high resolution 3D models.
Impacts of the chaotic behavior of the protein folding process in biology will be regarded.
Finally, the links between this established chaotic behavior and stochastic models in gene expression, mutation, or in Evolution, will be investigated.

\begin{small}
\bibliographystyle{plain}
\bibliography{mabase}
\end{small}



%

\end{document}